\pdfoutput=1
\documentclass[11pt, oneside]{amsart}

\usepackage{hyperref}
\hypersetup{%
	pdfpagemode={UseOutlines},
	bookmarksopen,
	pdfstartview={FitH},
	colorlinks,
	linkcolor={blue},
	citecolor={blue},
	urlcolor={blue}
}
\usepackage{amsmath, amssymb, amsthm, mathtools}
\usepackage[numbers]{natbib}

\numberwithin{equation}{section}

\newtheorem{teo}{Theorem}[section]

\newtheorem{lem}[teo]{Lemma}

\newtheorem{prop}[teo]{Proposition}
\newtheorem{conj}[teo]{Conjecture}

\newtheorem{corollary}[teo]{Corollary}
\theoremstyle{remark}
\newtheorem{remark}[teo]{Remark}

\newcommand{\R}{\mathbb{R}} 
\newcommand{\C}{\mathbb{C}} 
\renewcommand{\P}{\mathcal{P}} 
\renewcommand{\H}{\mathcal{H}} 
\newcommand{\bB}{{\mathbb{B}_n}}

\title[Wehrl-type entropy conjecture for $SU(N)$: cases of equality]{The Wehrl-type entropy conjecture for symmetric $SU(N)$ coherent states: cases of equality and stability}
\author{Fabio Nicola}
\address[Fabio Nicola]{Dipartimento di Scienze Matematiche, Politecnico di Torino, Corso Duca degli Abruzzi 24, 10129 Torino, Italy}
\email{fabio.nicola@polito.it}

\author{Federico Riccardi}
\address[Federico Riccardi]{Dipartimento di Scienze Matematiche, Politecnico di Torino, Corso Duca degli Abruzzi 24, 10129 Torino, Italy}
\email{federico.riccardi@polito.it}

\author{Paolo Tilli}
\address[Paolo Tilli]{Dipartimento di Scienze Matematiche, Politecnico di Torino, Corso Duca degli Abruzzi 24, 10129 Torino, Italy}
\email{paolo.tilli@polito.it}

\begin{document}	
	
	\keywords{Wehrl entropy, coherent states, representations, holomorphic polynomials, stability, Bergman spaces}
	\subjclass[2020]{81R30, 22E70, 30H20, 47N50, 49J40}

	\begin{abstract}
		\noindent Lieb and Solovej
		proved that, for the symmetric $SU(N)$ representations, the corresponding Wehrl-type entropy is minimized by symmetric coherent states. However, the uniqueness of the minimizers remained an open problem when $N\geq 3$. In this note, we complete the proof of the Wehrl entropy conjecture for such representations by showing that symmetric coherent states are, in fact, the only minimizers. We also provide an application to the maximum concentration of holomorphic polynomials and deduce a corresponding Faber-Krahn inequality. A sharp quantitative form of the bound by Lieb and Solovej is also proved. 
	\end{abstract}

    \maketitle
	
	\section{Introduction}
	In the late 1970s Wehrl introduced a notion of classical entropy for quantum density matrices and conjectured that the states with minimal classical entropy are exactly given by Glauber coherent states \cite[page 355]{wehrl_entropy}. In 1978 Lieb \cite{lieb_entropy} proved that this entropy is indeed minimized by Glauber coherent states, while later Carlen \cite{carlen} proved that these are the only minimizers. In his paper, Lieb also conjectured that an analogous result should hold for the irreducible representations of $SU(2)$: the desired bound was proved several years later by Lieb and Solovej \cite{lieb_solovej_SU(2)}, who further generalized the result to the symmetric representations of $SU(N)$ \cite{lieb_solovej_SU(N)}. The uniqueness of the extremizers was proved independently and simultaneously in \cite{frank_sharp_inequalities} and \cite{ortega} in the case $SU(2)$, while it remained open when $N \geq 3$. In this note, we complete the proof of the Wehrl entropy conjecture for the symmetric representations of $SU(N)$ by showing that the symmetric coherent states are the only minimizers. We anticipate here this result and refer to Section \ref{sec 2} for notation, terminology and proof. 
    
    Let $N \geq 2$, $M \geq 1$ be integers. Consider the irreducible representation of $SU(N)$ in the Hilbert space of totally symmetric tensor products given by  \[
    \H_M \coloneqq \big( \bigotimes\! ^M \C^N \big)_{\text{sym}}.
    \]
    The corresponding coherent states are the density operators in $\H_M$ (in fact, rank-one projections) of the form 
    \[
 |\otimes^M v \rangle \langle \otimes^M v|\quad  \textrm{with}\ v\in \C^N,\ |v|=1.
    \]
    With any density operator $\rho$ on $\H_M$ is associated a continuous  function $u:SU(N)\to[0,1]$, called the Husimi function of $\rho$. 
    
    The Lieb-Solovej inequality can be stated as follows. 
    \begin{teo}[Lieb-Solovej inequality \cite{lieb_solovej_SU(N)}]\label{th:Lieb-Solovej inequality}
		Let $\Phi \colon [0,1] \to \R$ be a convex function. Then, for any density operator $\rho$ on $\H_M$ we have
		\begin{equation}\label{eq:Lieb-Solovej inequality}
		\int_{SU(N)} \Phi( u(R) ) \, dR\leq 	\int_{SU(N)} \Phi( u_0(R) ) \, dR,
		\end{equation}
		where $u_0$ is the Husimi function of any coherent state.
	\end{teo}
Here $dR$ denotes the probability Haar measure on $SU(N)$. 
We now state our first result.
	\begin{teo}[full Wehrl's conjecture]\label{th:main theorem}
		Let $\Phi \colon [0,1] \to \R$ be a strictly convex function, and let $\rho$ be a density operator on $\H_M$ with Husimi function $u$. Then, the equality in \eqref{eq:Lieb-Solovej inequality} is achieved if and only if $\rho$ is a coherent state.
	\end{teo}
    
    In Section \ref{sec 3} we rephrase the Lieb-Solovej inequality \eqref{eq:Lieb-Solovej inequality}, and the above characterization of extremizers, in terms of concentration of holomorphic polynomials in $\C^{N-1}$.
 
Our uniqueness proof relies on the inequality \eqref{eq:Lieb-Solovej inequality}, as a tool. Nevertheless, our proof strategy is independent of the proof of the estimate \eqref{eq:Lieb-Solovej inequality} by Lieb and Solovej, which was based on a limiting argument (and hence did not lend itself to yield uniqueness). Therefore, our approach can hopefully be adapted to other contexts as well. In truth, a careful reader should easily infer the following somewhat surprising principle from the proof. \par\medskip\noindent 
{\bf Principle 1.} {\em
For Wehrl-type entropy bounds, from the sharp inequality for a sufficiently large class of convex functions, one automatically obtains the characterization of the cases of equality (and even a weak form of stability).}\par\medskip

In Section \ref{sec 4} (see Theorem \ref{teo main2}) we show that the same idea used in the proof of the uniqueness leads to a {\em sharp} quantitative form of  Lieb-Solovej inequality \eqref{eq:Lieb-Solovej inequality}. Again, from the proof, we can extract a principle that reads as follows. \par\medskip\noindent 
{\bf Principle 2.} {\em For Wehrl-type entropy bounds, stability is a consequence of the sharp inequality, combined with a suitable upper bound for the measure of small super-level sets of Husimi functions.}\par\medskip 
 We notice that stability results for concentration inequalities have been recently considered by several authors, in particular \cite{GGRT} and \cite{FNT} for Glauber coherent states (in the local and global form, respectively), \cite{GFOC} for holomorphic polynomials in one variable, and \cite{GKMR} for wavelet coherent states (local version). We emphasize that our approach also applies in all these contexts. 
    
  Then, as a further illustration of Principle 1, in Section \ref{sec 5} we consider the contractive inequality for weighted Bergman spaces on the unit ball of $\C^n$. This inequality, in dimension $n=1$, was originally proposed as a conjecture in \cite{LS21} and was subsequently solved (always in dimension $n=1$) in \cite{kulikov22}. In higher dimension, this problem is still open. However, we will prove that, also in this framework, the characterization of the extremizers follows for free once the inequality is obtained.  
\section{Characterization of the extremizers}\label{sec 2} 
	We adopt the notation from \cite{lieb_solovej_SU(N)}. Fix integers $N \geq 2$, $M \geq 1$ and consider the Hilbert space of totally symmetric tensor products of $N$-dimensional complex space, that is $\H_M \coloneqq ( \bigotimes^M \C^N )_{\text{sym}}$. Note that $\H_M$ is the image of $\bigotimes^M \C^N$ under the projection
	\begin{equation*}
		P_M (v_1 \otimes \cdots \otimes v_M) = \dfrac{1}{M!} \sum_{\sigma \in S_M} v_{\sigma(1)} \otimes \cdots \otimes v_{\sigma(M)},
	\end{equation*}
	where $S_M$ denotes the permutation group over $\{1,\ldots, M\}$. We denote by $\langle \psi|\phi\rangle$ the inner product of $\psi,\phi\in\H_M$, with the agreement that
	it is linear in the second argument. In the sequel, for greater clarity the dimension of
	$\H_M$ will be denoted by $\dim(\H_M)$, since its explicit form
	$\binom {N+M-1}{N-1}$ will not be needed.
	
	  We consider the representation of the group $SU(N)$ on $\bigotimes^M \C^N$ given by
	\begin{equation*}
		\begin{aligned}
			R(v_1 \otimes \cdots \otimes v_M) = (Rv_1) \otimes \cdots \otimes (Rv_M),\qquad R\in SU(N).
		\end{aligned}
	\end{equation*}
	Then $\H_M$ is an invariant subspace and the restriction of this representation to $\H_M$ is irreducible (see, e.g., \cite[Appendix A]{lieb_solovej_SU(N)}). 
	
	Let $\rho$ be a density operator (or density matrix) on $\H_M$, say
	\begin{equation}\label{eq:mixed state}
		\rho = \sum_{j=1}^{n} p_j |\psi_j \rangle \langle \psi_j |
	\end{equation}
	where $1 \leq n\leq \dim(\H_M)$,  $0 < p_j \leq 1$,  $\sum_{j=1}^n p_j = 1$ and the family $\{\psi_j\}_{j=1}^n \subset \H_M$ is such that $\langle \psi_j | \psi_k \rangle = \delta_{j,k}$. If there is just a single term in the sum, then $\rho$ is the projector onto a normalized vector $|\psi\rangle$ and is called a \emph{pure state}. In particular, if $|\psi\rangle = | \otimes^M v \rangle$ for some $v \in \C^N$, $|v|=1$, then $\rho$ is called a (symmetric) \emph{coherent state}.
	
	For a fixed vector $v_0 \in \C^N$, $|v_0|=1$, we consider the Husimi function $u \colon SU(N) \to \R$ associated with the density operator $\rho$,  defined as
	\begin{equation}\label{eq def husimi}
		u(R) \coloneqq \langle \otimes^M Rv_0 | \rho | \otimes^M Rv_0 \rangle = \sum_{j=1}^n p_j | \langle \otimes^M Rv_0 | \psi_j \rangle |^2.
	\end{equation}
	This is the function that is considered in the Lieb--Solovej inequality \eqref{eq:Lieb-Solovej inequality}. We observe that
    \[
    0\leq u(R)\leq 1\qquad \forall R\in SU(N)
    \]
    and (as is well known, cf. \cite[(10)]{lieb_solovej_SU(N)})
\begin{equation}\label{eq:integral of Husimi function}
		\dim (\H_M) \int_{SU(N)} u(R) \, dR = 1,
	\end{equation}
    where $dR$ denotes the probability Haar measure of $SU(N)$.
    
    \begin{remark}
    (i) The integral on the right-hand side of \eqref{eq:Lieb-Solovej inequality} can be written more explicitly (see \eqref{eq mu0} below) as 
    \begin{align*}
\int_{SU(N)} \Phi( u_0(R) ) \, dR&=\int_0^1 \Phi\big((1-s^{1/(N-1)})^M\big) \, ds
\\ &=\lim_{t\to0^+} \Phi(t)+\int_0^1 \Phi'(t) \big(1-t^{1/M}\big)^{N-1}\, dt;
    \end{align*}
    (we will see that the Husimi function $u$ achieves the values $0$ and $1$ on subsets of measure $0$, so that the values --- and the possible discontinuity --- of $\Phi$ at $0$ and $1$ are irrelevant). 
    \par\medskip
    (ii) Multiplying $\psi\in \H _M$, with $\|\psi\|=1$, by a phase factor, does not affect the density operator $
    \rho=|\psi\rangle \langle \psi|$. Hence, in the case of pure states, the integral on the left-hand side of \eqref{eq:Lieb-Solovej inequality} can be regarded as a function $G:P(\H_M)\to\mathbb{R}$ where $P(\H_M)$ is the projective space over $\H_M$, that is 
    \[
    G([\psi])=\int_{SU(N)} \Phi(|\langle \otimes^M Rv_0|\psi\rangle|^2)\, dR,\qquad [\psi]\in P(\H_M),
    \]
    (with the agreement that  $\|\psi\|=1$).  
    Moreover, $G$ is invariant under the action of the group $PSU(N):=SU(N)/\{\lambda_k I: k=0,\ldots,N-1\}$, with $\lambda_k=e^{2\pi ik/N}$.  
    Hence, Theorems \ref{th:Lieb-Solovej inequality} and \ref{th:main theorem} tell us that $G$ achieves its maximum only at the points of the orbit \[
\{[\otimes ^M v]:\ v\in\C^{N},\ |v|=1\}\subset P(\H_M). 
\]
This orbit is a famous rational algebraic variety, known as the Veronese variety; cf. \cite[Section 11.3]{fulton_harris}.
\end{remark}
We can now prove Theorem \ref{th:main theorem}. \par\medskip\noindent 
\textit{Proof of Theorem \ref{th:main theorem}.}
	The fact that coherent states achieve equality in \eqref{eq:Lieb-Solovej inequality} is already contained in Theorem \ref{th:Lieb-Solovej inequality}.
	To prove that equality occurs only for such states we rely on the following lemma. We denote by $|A|$ the (normalized Haar) measure of a Borel subset $A\subset SU(N)$.
	\begin{lem}\label{lem lemma 23}
		Let $\Phi \colon [0,1] \to \R$ be a convex function, let $\rho$ be a density operator on $\H_M$ with Husimi function $u$, and let $u_0$ be the Husimi function of a coherent state. Moreover, let $\mu(t) = |\{u > t\}|$ and $\mu_0(t) = |\{u_0 > t\}|$ ($t\geq0$) denote the corresponding distribution functions,  and let
		\[
		T := \sup_{R \in SU(N)} u(R).
		\]
		Then, the following stability inequality
		\begin{equation}\label{eq:stability estimate}
			\int_{SU(N)} \Phi(u_0(R)) \, dR - \int_{SU(N)} \Phi(u(R)) \, dR \geq \int_T^1 \left(\Phi'(t) - \Phi'_-(T)\right) \mu_0(t) \, dt
		\end{equation}
		holds, where $\Phi'_-$ denotes the left derivative of $\Phi$.
	\end{lem}
	
	\begin{proof}
		When $T=1$ \eqref{eq:stability estimate} becomes the inequality \eqref{eq:Lieb-Solovej inequality}, so we can assume that $0<T<1$. We consider the decomposition $\Phi(t) = \Phi_1(t) + \Phi_2(t)$, where
		\begin{equation*}
			\Phi_1(t) = \begin{cases}
				\Phi(t) & 0 <t \leq T \\
				\Phi'_-(T)(t-T) + \Phi(T) & T < t \leq 1.
			\end{cases}
		\end{equation*}
		Then, by applying  $\eqref{eq:Lieb-Solovej inequality}$ to
		$\Phi_1$, we have
		\begin{align*}
			& \int_{SU(N)} \Phi(u_0(R)) \, dR - \int_{SU(N)} \Phi(u(R)) \, dR \\
			=& \int_{SU(N)} \Phi_1(u_0(R)) \, dR - \int_{SU(N)} \Phi_1(u(R)) \, dR \\
			+& \int_{SU(N)} \Phi_2(u_0(R)) \, dR - \int_{SU(N)} \Phi_2(u(R)) \, dR \\
			\geq& \int_{SU(N)} \Phi_2(u_0(R)) \, dR - \int_{SU(N)} \Phi_2(u(R)) \, dR  & \\
			=& \int_T^1 (\Phi'(t) - \Phi'_-(T)) \mu_0(t) \, dt - \int_T^1 (\Phi'(t) - \Phi'_-(T)) \mu(t) \, dt \\
			=& \int_T^1 (\Phi'(t) - \Phi'_-(T)) \mu_0(t) \, dt, &
		\end{align*}
		where in the last passage we have used the fact that $\mu(t)=0$ when $t\in (T,1]$.
	\end{proof}	
	Now, consider the right-hand side of \eqref{eq:stability estimate}. If $T<1$ and $\Phi$ is strictly convex, we have $\Phi'(t) - \Phi'_-(T) > 0$ in $(T,1)$ and also $\mu_0(t) > 0$ in the same interval, therefore $T<1$ implies
	\begin{equation*}
		\int_{SU(N)} \Phi(u_0(R)) \, dR > \int_{SU(N)} \Phi(u(R)) \, dR,
	\end{equation*}
	which means that equality can hold only if $T=1$.
	\begin{remark}
		Notice that the same argument applies, more generally, if $\Phi$ is convex and, for every $\varepsilon\in (0,1)$, $\Phi$ is not an affine function on the interval $(1-\varepsilon,1)$. Consequently, the conclusion of Theorem \ref{th:main theorem} (and Corollary \ref{cor 1} below) extend to any convex function $\Phi$ (not necessarily strictly convex) having this property.
	\end{remark}
	
	To conclude the proof of the theorem, we need to prove that
	\begin{equation}\label{eq implicazione T=1}
		T = 1 \Longrightarrow\rho = |\otimes^M v \rangle \langle \otimes^M v |
	\end{equation}
	for some $v \in \C^N$, $|v|=1$.
    To this end, with the notation in \eqref{eq:mixed state}, we observe that if there exists $R\in SU(N)$ such that 
    \[
    \sum_{j=1}^n p_j | \langle \otimes^M Rv_0 | \psi_j \rangle |^2=1
    \]
    then
    \[
|\langle \otimes^M Rv_0 | \psi_j \rangle |=1
    \]
   for every $j=1,\ldots,n$, and therefore
    \[
\psi_j =e^{i\theta_j} \otimes^M Rv_0
    \]
    for some $\theta_j\in\R$. Since the $\psi_j's$ are pairwise orthogonal, it follows that $n=1$ and therefore $\rho$ is a coherent state. This concludes the proof of Theorem \ref{th:main theorem}. 

    \begin{remark} As the reader will have noticed, the above argument is very general and shows that, for Wehrl-type entropy bounds, the uniqueness of the extremizers (for strictly convex functions) follows essentially from the bound itself once the latter is known to hold \textit{for every convex function}. 
    \end{remark} 

    \section{Applications to the maximum concentration of holomorphic polynomials}\label{sec 3}
Theorems \ref{th:Lieb-Solovej inequality} and \ref{th:main theorem} have an interesting consequence on the maximum concentration of holomorphic polynomials  in $\C^{N-1}$. The connection with complex analysis is provided by the fact that the natural classical phase space associated with $\H_M$ is the complex projective space $\C P^{N-1}$, which is a complex manifold. In this section we discuss this application because of its great intrinsic interest, and also as a preparation to the stability analysis in the next section, that --- unlike the above proof of Theorem \ref{th:main theorem}  --- will exploit the complex structure of $\C P^{N-1}$ (through Lemma \ref{lem lemma paolo} below).

\subsection{Global estimates}	The starting point of this discussion is the observation that the Husimi function defined in \eqref{eq def husimi} is constant on the fibers of the fibrations
        \begin{equation*}
		\begin{aligned}
			SU(N)          &\longrightarrow \phantom{||||||||||||} S^{2N-1}                    &\longrightarrow & \phantom{||||||}\C P^{N-1} \\
			R \phantom{aa} &\longmapsto     Rv_0 = (z_1, \dots, z_N)    &\longmapsto     &\, [z_1 \colon \cdots \colon z_N],
		\end{aligned}
	\end{equation*}
	which are $SU(N-1)$ and $U(1)$, respectively.
	Hence the Husimi function can be regarded as a function on $S^{2N-1}$ (the $2N-1$ dimensional real sphere in $\C^N$) or even on the complex projective space $\C P^{N-1}$.
	As a consequence, we have the following equalities of integrals with respect to normalized $SU(N)$-invariant measures
	\begin{equation*}
		\int_{SU(N)} \Phi (u(R)) \, dR = \int_{S^{2N-1}} \Phi (u(v)) \, dv = \int_{\C P^{N-1}} \Phi(u(z)) \, dz,
	\end{equation*}
	where $\Phi \colon [0,1] \to \R$ is an arbitrary Borel function for which the integrals make sense (with some abuse of notation, we denote the Husimi function with the same symbol $u$, regardless of the domain being considered).
	
	To make the last integral more explicit, we consider the affine chart of $\C P^{N-1}$ where $z_1 \neq 0$ (homogeneous coordinates), with coordinates $z' = (z_2/z_1,\dots , z_N / z_1) \in \C^{N-1}$. We have
	\begin{align*}
		\int_{\C P^{N-1}} \Phi(u(z)) \, dz &= \int _{\C P^{N-1}} \Phi \left( \sum_{j=1}^n p_j \lvert \langle \otimes^M \frac{z}{|z|} | \psi_j \rangle \rvert^2 \right)\, dz\\
		&= \int_{\C^{N-1}} \Phi \left( \frac{\sum_{j=1}^n p_j |\langle \otimes^M (1,z') | \psi_j \rangle|^2}{(1+|z'|^2)^M} \right) d\nu(z'),
	\end{align*}
	where
	\[
	d \nu(z') := \frac{c_N}{(1+|z'|^2)^{N}} \, dA(z')
	\]
	is a probability measure on $\C^{N-1}$, with $c_N=\frac{(N-1)!}{\pi^{N-1}}$ (here $dA(z')$ is the Lebesgue measure on $\C^{N-1}$). We point out that $\langle \otimes^M (1,z') | \psi_j \rangle$ is an antiholomorphic polynomial of degree at most $M$ in $\C^{N-1}$.

	Suppose now that $\rho$ is a pure state. Then its Husimi function, in the chart $\C^{N-1}$, has the form
	\[
	\frac{|F(z')|^2}{(1+|z'|^2)^M}
	\]
	for some holomorphic polynomial $F$ of degree at most $M$ in $\C^{N-1}$. Moreover, it is easy to see that every  polynomial of that type, suitably normalized, occurs in this way.  Observe that $F(z')=e^{i\theta}$, for some $\theta\in\R$, if $\rho=| \otimes^M v\rangle \langle \otimes^M v|$ with $v=(1,0,\ldots,0)$. More generally, the Husimi function of the state $| \otimes^M v\rangle \langle \otimes^M v|$, with $v\in\C^N$, $|v|=1$, is given by the function
		\begin{equation}\label{eq uzero}
		u_0(z)=\frac{|\langle v|z\rangle_{\C^{N}}|^{2M}}
		{|z|^{2M}}.
	\end{equation}
	This suggests considering the reproducing kernel Hilbert space $\P_M$ of holomorphic polynomials $F(z')$ of degree at most $M$ in $\C^{N-1}$, equipped with the norm
	\begin{equation*}
		\|F\|_{\P_M}^2 \coloneqq \dim(\H_M) \int_{\C^{N-1}} \dfrac{|F(z')|^2}{(1+|z'|^2)^M} \, d\nu(z')=\dim(\H_M) \int_{\C P^{N-1}}u(z)\, dz,
	\end{equation*}
	where we set
	\begin{equation}\label{eq def f}
		u(z)=\frac{|F(z')|^2}{(1+|z'|^2)^M}=\frac{|z_1^M F(z_2/z_1,\ldots, z_N/z_1)|^2}{|z|^{2M}}.
	\end{equation}
	Observe that $\|1\|_{\P_M} = 1$, as follows immediately from \eqref{eq:integral of Husimi function}.

	 Theorems \ref{th:Lieb-Solovej inequality} and \ref{th:main theorem} then clearly lead to the following result (the case $N=2$ was addressed in \cite{frank_sharp_inequalities,ortega,lieb_solovej_SU(2)}). 
\begin{corollary}\label{cor 1}
        Let $F\in\P_M$, with $\|F\|_{\P_M}=1$. Then, for every convex function $\Phi:[0,1]\to\R$
we have 
\begin{equation}\label{eq cor 3.1}
\int_{\C^{N-1}} \Phi \left( \frac{ |F(z')|^2}{(1+|z'|^2)^M} \right) d\nu(z')\leq \int_{\C^{N-1}} \Phi \left( \frac{ 1}{(1+|z'|^2)^M} \right) d\nu(z').
\end{equation}
Moreover, if $\Phi$ is strictly convex, equality occurs if and only if 
\[
F(z')=\langle v| (1,z')\rangle_{\C^N}^M,\qquad z'\in\C^{N-1}
\]
for some $v\in\C^N$, with $|v|=1$. 
\end{corollary}
We conclude this section with the following remark. The implication \eqref{eq implicazione T=1} --- which is in fact an equivalence --- can be rephrased as a property of the polynomials $F$ in $\P_M$ as follows. We provide a proof that prescinds from the fact that $F$ comes from the Husimi function of a pure state in $\H_M$. 
	\begin{prop}\label{lem 25}
		Let $F\in \P_M$ and let $u(z)$ be the corresponding function in \eqref{eq def f}. Assume that $\|F\|_{\P_M}=1$ or, equivalently, that
		\begin{equation}\label{eq norm}
			\dim(\H_M) \int_{\C P^{N-1}}u(z)\, dz=1.
		\end{equation}
		Then
		\begin{equation}\label{eq stima uv}
			u(v)\leq 1\qquad \forall v\in\C P^{N-1}.
		\end{equation}
		Moreover, equality occurs at a point of $\C P^{N-1}$ represented by $v\in \C^N$, $|v|=1$, if and only if
		\[
		u(z)=\frac{|\langle v|z\rangle_{\C^{N}}|^{2M}}
		{|z|^{2M}}.
		\]
	\end{prop}
	\begin{proof}
		First, we address the desired estimate at the point $[1:0:\ldots:0]\in\C P^{N-1}$,
		which amounts to proving
		the estimate
		\begin{equation}\label{eq Fzero}
			|F(0)|\leq 1,
		\end{equation}
		and the fact that equality occurs if and only if $F$ is constant, namely $F(z')=e^{i\theta}$ for some $\theta\in\R$. Observe that the corresponding function $u(z)$ in \eqref{eq def f} in that case is given by $u(z)=|z_1|^{2M}/|z|^{2M}$.
		
		The proof of the inequality \eqref{eq Fzero} is standard, using the subharmonicity of $|F(z')|^2$ and the polar coordinates in $\C^{N-1}$. Also, one sees that equality can occur only if $|F|^2$ is, in fact, harmonic, that forces $F$ to be constant. This is a consequence of the equality
        \[
\Delta |F|^p=p^2|F|^{p-2}|\partial F|^2,
        \]
        which holds for every $p>0$ on the set where $F\not=0$. 
		
		Now, consider a point of $\C P^{N-1}$ represented by some $v\in\C^N$, $|v|=1$. Let $v=R(1,0,\ldots,0)$, for some $R\in SU(N)$. Then, if $u$ is a function as in the statement, we see that $u(R\, \cdot)$ has still the same form as in \eqref{eq def f} (for a new polynomial $F$ in $\P_M$), and satisfies the same normalization \eqref{eq norm}. Therefore, applying the result already proved to the function $u(R \cdot)$ we see that $u(v)\leq 1$
		and that equality occurs if and only if
		\[
		u(Rz)=\frac{|z_1|^{2M}}{|z|^{2M}},
		\]
		that is
		\[
		u(z)=\frac{|\langle v|z\rangle_{\C^{N}}|^{2M}}
		{|z|^{2M}}.
		\]
		This concludes the proof.
	\end{proof}
\subsection {Local estimates} Arguing as in \cite[Section 5] {frank_sharp_inequalities} we easily obtain a corresponding local estimate --- also known as Faber-Krahn type inequality (see \cite{frank_sharp_inequalities,ortega} for the case $N=2$, \cite{NT} for the analogous result for functions in the Fock space, and \cite{RT} for functions in Bergman spaces). 

Setting  $f(z'):=|F(z')|^2(1+|z'|^2)^{-M}$, $f_0(z'):=(1+|z'|^2)^{-M}$ and denoting by $f^\ast(s)$ and $f_0^\ast(s)$ their decreasing rearrangements on the interval $[0,1]$ (recall that the measure of $\C P^{N-1}$ is normalized to $1$) the estimate  \eqref{eq cor 3.1} is equivalent to 
\[
\int_0^1\Phi(f^\ast(\tau))\, d\tau\leq \int_0^1 \Phi(f_0^\ast(\tau))\, d\tau
\]
for every convex function $\Phi:[0,1]\to\R$. By the Hardy-Littelwood-Polya majorization theory \cite[subsections 249 and 250]{HLP}, the latter estimate is equivalent to 
\[
\int_0^s f^\ast(\tau)\, d\tau\leq \int_0^s f_0^\ast(\tau)\, d\tau\qquad \forall s\in[0,1], 
\]
which implies the following result. 
\begin{corollary}\label{cor 2}
Let $F\in\P_M$, with $\|F\|_{\P_M}=1$ and let $\Omega\subset\C^{N-1}$ be a Borel subset of measure $\nu(\Omega)>0$. Then
\begin{equation}\label{eq cor 3.2}
\int_{\Omega} \frac{ |F(z')|^2}{(1+|z'|^2)^M} \,  d\nu(z')\leq \int_{\Omega^\ast} \frac{ 1}{(1+|z'|^2)^M} \, d\nu(z'),
\end{equation}
where $\Omega^\ast$ is the Euclidean ball in $\C^{N-1}$ of center $0$ and measure $\nu(\Omega^\ast)=\nu(\Omega)$.

Moreover, equality occurs if  
\[
F(z')=\langle v| (1,z')\rangle_{\C^N}^M,\qquad z'\in\C^{N-1}
\]
for some $v\in\C^N$, with $|v|=1$, and $\Omega$ is a super-level set of the function 
\[
\C^{N-1}\ni z'\mapsto |\langle v| (1,z')\rangle_{\C^N}|^{2M}(1+|z'|^2)^{-M}
\]
(hence, regarded as a subset of $\C P^{N-1}$, $\Omega$ is a geodesic ball of center $v$).
\end{corollary}

    \section{The bound in quantitative form}\label{sec 4}
    In this section we show that the same proof strategy of Theorem \ref{th:main theorem}, when combined with a suitable upper bound for the distribution function of the Husimi function, yields, in fact, a sharp quantitative form of the bound \eqref{eq:Lieb-Solovej inequality}. This illustrates Principle 2 (see Introduction). 

We adopt the notation of Section \ref{sec 2}. Hence, we denote by $\rho$ a density operator on the space $\H_M$ and by $u$ its Husimi function. Let $D[\rho]$ be the distance, in the trace norm, between $\rho$ and the subset of coherent states, that is 
\[
D[\rho]:=\inf_{v\in \C^N, |v|=1}\| \rho- | \otimes^M v\rangle \langle \otimes^M v|\|_1.
\]
Then, we have the following sharp bound. 
\begin{teo}\label{teo main2}
For every strictly convex function $\Phi:[0,1]\to\R$, there exists a constant $c>0$ such that, for every density operator $\rho$ in $\H_M$ we have 
\begin{equation}\label{eq stabilita}
    \int_{SU(N)} \Phi( u_0(R) ) \, dR- \int_{SU(N)} \Phi( u(R) ) \, dR\geq c D[\rho]^2,
		\end{equation}
		where $u$ is the Husimi function of $\rho$ and $u_0$ is the Husimi function of any coherent state.
	\end{teo}
    It follows from the proof that the constant $c$ in \eqref{eq stabilita} is explicit, in the sense that it is not
obtained by a compactness argument.

In order to prove this result we first observe that, by a general argument (which is detailed in \cite[Proposition 2.2]{FNT}) we have 
\begin{equation}\label{eq drho}
D[\rho]^2\leq 4(1-T)
\end{equation}
where $T=\sup_{R\in SU(N)} u(R)$ (with equality if $\rho$ is a pure state). 

With the intention to enhance the estimate in Lemma \ref{lem lemma 23}, we therefore begin with an inspection of the distribution functions $\mu(t)=|\{u>t\}|$ and $\mu_0(t)=|\{u_0>t\}|$. 

By an explicit computation, working in the affine chart $\C^{N-1}$ as in the previous section, using \eqref{eq uzero} and polar coordinates, we easily see that, for $0\leq t\leq 1$,
\begin{equation}\label {eq mu0}
\mu_0(t)=\frac{(N-1)!}{\pi^{N-1}}\int_{|z'|<R}\frac{1}{(1+|z'|^2)^N}\, dA(z')=(1-t^{1/M})^{N-1},
\end{equation}
where we set $R=\sqrt{\frac{1-t^{1/M}}{t^{1/M}}}$  ($0\leq t\leq 1$).

Concerning $\mu(t)$, we need the following upper bound. 
\begin{lem}\label{lem lemma paolo}
    For every $t_0\in (0,1)$ there exists $T_0\in(t_0,1)$ and $C_0>0$ such that, for $T\in [T_0,1]$,
    \begin{equation}\label{eq lemma paolo}
    \mu(t)\leq (1+C_0(1-T))\mu_0(t/T),\qquad \forall t\in [t_0,T]. 
    \end{equation}
    \end{lem}
    \begin{proof}
        We regard $u$ as a function on $\C P^{N-1}$, and therefore on the affine chart $\C^{N-1}$, as in the previous section. Then $\mu(t)$ represents the $\nu$-measure in $\C^{N-1}$ of the super-level set 
        \[
        \{z'\in\C^{N-1}:\ |F(z')|^2(1+|z'|^2)^{-M}>t\}
        \]
        for some $F\in\P_M$, $\|F\|_{\P_M}=1$ --- if $\rho$ is a pure state --- or more generally of the set \[
        \big\{z'\in\C^{N-1}:\ \sum_{j=1}^n p_j|F_j(z')|^2(1+|z'|^2)^{-M}>t\big\},
        \]
        where the $F_j$'s are orthonormal in $\P_M$. Now, an estimate analogous to \eqref{eq lemma paolo} was already proved in \cite[Lemma 2.1]{GGRT} (see also \cite[page 823]{GGRT} for the multidimensional case and \cite[Lemma 2.6]{FNT} for the case of density operators) for functions in the Fock space, that is, when the above weight $(1+|z'|^2)^{-M}$ is replaced by $e^{-\pi|z'|^2}$. A careful inspection of the proof of \cite[Lemma 2.1]{GGRT} shows that each step of that argument can be adapted to the present situation, and this leads to the desired bound \eqref{eq lemma paolo}. We omit the details because several variations on the theme of \cite[Lemma 2.1]{GGRT} have already appeared in the literature (see, e.g., \cite{GFOC} for holomorphic polynomials in one variable and \cite{GKMR} for analogous results in Bergman spaces), and this machinery can therefore be considered well known to experts. In addition, the bound \eqref{eq lemma paolo} was also recently proved in \cite{GFOCprogress} (private communication). 
    \end{proof}
    Now we come to the proof of Theorem \ref{teo main2}. Let $\tau_1\in(0,1)$ be a constant that will be chosen later. By the same argument as in Lemma \ref{lem lemma 23}, with $T$ replaced by $\tau_1$, we obtain the estimate 
    \begin{align*}
\int_{SU(N)} \Phi( u_0(R) ) \, dR- \int_{SU(N)} &\Phi( u(R) ) \, dR
\\ &\geq \int_{\tau_1}^1 (\Phi'(t)-\Phi'_-(\tau_1))(\mu_0(t)-\mu(t))\, dt. 
    \end{align*}
    Hence, by \eqref{eq drho} we see that it is  sufficient to prove that 
    \[
\int_{\tau_1}^1 (\Phi'(t)-\Phi'_-(\tau_1))(\mu_0(t)-\mu(t))\, dt\geq c(1-T)
    \]
    for some $c>0$, for a suitable choice of $\tau_1$. We can also suppose that $T$ is sufficiently close to $1$, because otherwise the desired result follows from Lemma \ref{lem lemma 23}.   

Let $t_0=1/2$ and let $T_0$ be the corresponding threshold in Lemma \ref{lem lemma paolo}. We see that it is sufficient to prove that there exist $\tau_1,\tau_2,\tau_3$, with $T_0\leq \tau_1<\tau_2<\tau_3<1$, and $\varepsilon>0$ such that 
\begin{equation}\label{eq diff1}
\mu_0(t)-(1+C_0(1-T))\mu_0(t/T)\geq 0\qquad \textrm{for}\ \tau_1\leq t\leq T
\end{equation}
and 
\begin{equation}\label{eq diff2}
\mu_0(t)-(1+C_0(1-T))\mu_0(t/T)\geq \varepsilon(1-T)\qquad  \textrm{for}\ \tau_1\leq t\leq\tau_2<\tau_3\leq  T.
\end{equation}
Indeed, from Lemma \ref{lem lemma paolo},  \eqref{eq diff1} and \eqref{eq diff2} we deduce that 
\begin{align*}
    \int_{\tau_1}^1 (\Phi'(t)&-\Phi'_-(\tau_1))(\mu_0(t)-\mu(t))\, dt\\
    &\geq \int_{\tau_1}^1 (\Phi'(t)-\Phi'_-(\tau_1))\big(\mu_0(t)-(1+C_0(1-T))\mu_0(t/T)\big)\, dt 
    \\
    &\geq \int_{\tau_1}^{\tau_2} (\Phi'(t)-\Phi'_-(\tau_1))\big(\mu_0(t)-(1+C_0(1-T))\mu_0(t/T)\big)\, dt \\
    &\geq c(1-T),
\end{align*}
with $c=\varepsilon \int_{\tau_1}^{\tau_2}(\Phi'(t)-\Phi'_-(\tau_1))\, d\tau>0$. 

It remains to prove \eqref{eq diff1} and \eqref{eq diff2}. Observe that, for $1/2\leq t\leq T\leq 1$, we have 
\begin{align*}
\mu_0(t)-(1+C_0&(1-T))\mu_0(t/T)\\
&=\int_T^1\Big(-\frac{t}{\tau^2}\mu'_0(t/\tau)(1+C_0(1-\tau))-C_0\mu_0(t/\tau) \Big)\,d\tau \\
&\geq \int_T^1\Big(-\frac{1}{2}\mu'_0(t/\tau)-C_0\mu_0(t/\tau) \Big)\,d\tau,
\end{align*}
where we used that $\mu'_0\leq0$. Then, \eqref{eq diff1} and \eqref{eq diff2}  follow by observing that, setting 
\[
\phi(r):=-\frac{1}{2}\mu'_0(r)-C_0\mu_0(r)\qquad r\in [0,1],
\]
we have 
\[
\phi(r)> 0\qquad\textrm{for}\ r_1\leq r< 1
\]
if $r_1\in(0,1)$ is sufficiently close to $1$, and therefore, for every $r_2\in(r_1,1)$ we also have
\[
\phi(r)\geq \varepsilon\qquad\textrm{for}\ r_1\leq r\leq r_2
\]
for some $\varepsilon>0$. 
This is a consequence of the fact that $\mu_0$ in \eqref{eq mu0} is decreasing and (as a function on [0,1]) vanishes to finite order at $1$. 
This concludes the proof of Theorem \ref{teo main2}. 
\begin{remark}
From the above proof it follows that \eqref{eq diff1} and \eqref{eq diff2} hold in fact for every triple $0<\tau_1<\tau_2<\tau_3<1$, with $\tau_1$ sufficiently close to $1$. Hence Theorem \ref{teo main2} extends to any convex function $\Phi$ (not necessarily strictly convex) such that, for every $\varepsilon\in (0,1)$, $\Phi$ is not an affine function on the interval $(1-\varepsilon,1)$.
\end{remark}
\begin{remark}\label{rem unicita faber-krahn} With the same notation of the proof of Theorem \ref{teo main2}, by Lemma \ref{lem lemma paolo} and the argument that led to \eqref{eq diff1} we see that, if $u$ is the Husimi function of a state that is not a coherent state, hence $T<1$ by \eqref{eq implicazione T=1}, we have $\mu(t)<\mu_0(t)$ for $t\in [\tau_1,1]$ for a suitable constant $\tau_1\in (0,1)$. As a consequence, denoting by $u^\ast(s)$ and $u_0^\ast(s)$, with $s\in [0,1]$, the decreasing rearrangements of $u$ and $u_0$, respectively, and setting $s_1=\mu_0(\tau_1)$, we have $s_1\in (0,1)$ and 
\[
u^\ast(s)<u_0^\ast (s)\qquad s\in [0,s_1].
\]
This implies that, in Corollary \ref{cor 2}, if $0<\nu(\Omega)<s_1$, equality occurs in \eqref{eq cor 3.2} if and only if 
\[
F(z')=\langle v| (1,z')\rangle_{\C^N}^M,\qquad z'\in\C^{N-1}
\]
for some $v\in\C^N$, with $|v|=1$, and $\Omega$ is a super-level set of the function 
\[
\C^{N-1}\ni z'\mapsto |\langle v| (1,z')\rangle_{\C^N}|^{2M}(1+|z'|^2)^{-M}.
\]
\end{remark}
\section{Extremizers of the contractive estimates in weighted Bergman spaces}\label{sec 5}
In this section we illustrate Principle 1 (see Introduction) in the context of contractive inequalities for weighted Bergman spaces on the unit ball of $\C^n$. 

Let $\bB$ be the unit ball in $\C^n$, $n\geq 1$, and let $dv(z)$, with $z\in\bB$, be the Lebesgue measure.
 Consider the hyperbolic measure 
\[
dv_g(z)=\frac{dv(z)}{(1-|z|^2)^{n+1}}.
\]
For $0<p<\infty$, $\alpha>n$, let $A^p_\alpha$ be the space of holomorphic functions $f$ on $\bB$ such that 
\[
\|f\|^p_{A^p_\alpha}:=c_n\int_\bB |f(z)|^p(1-|z|^2)^\alpha dv_g(z)<\infty, 
\]
with $c_\alpha=\frac{\Gamma(\alpha)}{\alpha!\Gamma(\alpha-n)}$, so that $\|1\|_{A^p_\alpha}=1$ for every admissible value of $p$ and $\alpha$. 

The following inequality was conjectured in \cite{LS21}, in dimension $n=1$.
\begin{conj}\label{thm li_su}
    Let $\Phi:[0,1]\to\mathbb{R}$ be a convex function. For every $p\in(0,\infty)$, $\alpha>n$ and $f\in A^p_\alpha$, with $\|f\|_{A^p_\alpha}=1$, we have 
    \begin{equation}\label{eq li_su}
        \int_\bB \Phi(|f(z)|^p(1-|z|^2)^\alpha)\,dv_g(z) \leq \int_\bB \Phi((1-|z|^2)^\alpha)\,dv_g(z). 
    \end{equation}
\end{conj}
This conjecture was recently proved in \cite[Theorem 1.2 and Remark 4.3]{kulikov22} in dimension $n=1$, where a characterization of the extremizers is also provided. 
In dimension $n\geq 2$ this inequality is still open (it is known that the isoperimetric conjecture, that is, that isoperimetric subsets of $\bB$ are geodesic balls, implies Conjecture \ref{thm li_su}; see \cite{ortega,li_su}). 
The following result gives a characterization of the extremizers in arbitrary dimension, provided that the above inequality holds true. 

\begin{teo}
    Let $p\in(0,\infty)$, $\alpha>n$. Suppose that for every $f\in A^p_\alpha$, with $\|f\|_{A^p_\alpha}=1$, and for every $\Phi:[0,1]\to\R$ convex the inequality \eqref{eq li_su} holds true. 
    
 Then, equality occurs in \eqref{eq li_su} if 
\begin{equation}\label{eq f extre}
    f(z)=\frac{e^{i\theta}(1-|w|^2)^{\alpha/p}}
    {(1-\langle w|z\rangle_{\C^n})^{2\alpha/p}},\qquad z\in\bB
\end{equation}
for some $w\in\bB$ and $\theta\in \R$. 

Moreover, assuming $\Phi$ strictly convex, if the integral on the right-hand side of \eqref{eq li_su} is finite and for some $f\in A^p_\alpha$, with $\|f\|_{A^p_\alpha}=1$, equality occurs in \eqref{eq li_su}, then $f$ has the form in \eqref{eq f extre} for some $w\in\bB$ and $\theta\in \R$. 
\end{teo}
\begin{proof}
 It is well known (see \cite[page 5]{zhu_holomorphic}) that, for every $w\in\bB$ there exists an automorphism $\varphi_w$ of $\bB$ (that is, a biholomorphic mapping $\bB\to\bB$ preserving the Bergman metric, and therefore the Bergman measure $dv_g$), with $\varphi_w(w)=0$, and satisfying 
\begin{equation}\label{eq phiw}
    1-|\varphi_w(z)|^2=\frac{(1-|w|^2)(1-|z|^2)}{|1-\langle w|z\rangle_{\C^n}|^2}\qquad z,w\in\bB. 
\end{equation}
As a consequence, the map $f\mapsto f_w$, given by 
\[
f_w(z):= f(\varphi_w(z))\frac{(1-|w|^2)^{\alpha/p}}
    {(1-\langle w|z
    \rangle_{\C^n})^{2\alpha/p}}\qquad z\in\bB,
\]
is an isometric isomorphism of $A^p_\alpha$. The first part of the statement follows from \eqref{eq phiw}, since equality (of course) occurs in \eqref{eq li_su} if $f\equiv 1$. 

Suppose now $\Phi$ is strictly convex and that the integral on the right-hand side of $\eqref{eq li_su}$ is finite. Let $f\in A^p_\alpha$, with $\|f\|_{A^p_\alpha}=1$, achieve equality in $\eqref{eq li_su}$. Set 
\[
T:=\sup_{z\in \bB} \big(|f(z)|^p (1-|z|^2)^{\alpha}\big).
\]
It is well known that $T\in (0,1]$ and that the above supremum is achieved, since $f(z) (1-|z|^2)^{\alpha/p}\to 0$ as $|z|\to 1^-$ (see \cite[Theorem 2.1]{zhu_holomorphic} and the subsequent discussion). We can then argue as in the proof of Lemma \ref{lem lemma 23} (observing that the integrals that arise from the decomposition $\Phi=\Phi_1+\Phi_2$ are finite) and we obtain
\begin{align*}
\int_\bB \Phi((1-|z|^2)^\alpha)\,dv_g(z)-\int_\bB &\Phi(|f(z)|^p(1-|z|^2)^\alpha)\,dv_g(z)\\
&\geq \int_T^1 (\Phi(t)-\Phi'_-(T))\mu_0(t)\, dt,
\end{align*}
where now 
\[
\mu_0(t):=v_g(\{z\in\bB:\ (1-|z|^2)^{\alpha}>t\}). 
\]
Hence, it is sufficient to prove that $T=1$ only for the functions $f$ in \eqref{eq f extre}, that is, that equality occurs in the pointwise estimate  
\[
|f(z)|^p (1-|z|^2)^{\alpha}\leq 1
\]
at some point $w\in\bB$, only if $f$ has the form in \eqref{eq f extre}. Again, using the above transformation $f\mapsto f_w$ we see that we can assume $w=0$. And if $f(0)=1$, arguing as in the first part of the proof of Proposition \ref{lem 25}, it follows that $f$ is constant ($f(z)=e^{i\theta}$ for some $\theta\in\R$). 
\end{proof}

	\section*{Acknowledgments} 
    We would like to thank J. Ortega-Cerd\`a for informing us about their work in progress \cite{GFOCprogress} about the concentration of holomorphic polynomials. 
    
    F.~N.~ is a Fellow of the {\em Accademia delle Scienze di Torino} and a member of the {\em Societ\`a Italiana di Scienze e Tecnologie Quantistiche (SISTEQ)}.
	

\end{document}